\documentclass[a4paper,12pt]{article}
\usepackage{latexsym,amssymb,amsfonts,amsmath,amsthm}
\usepackage[dvips]{graphicx}
\usepackage{color}
\usepackage{ulem}

\newtheorem{theorem}{Theorem}
\newtheorem{lemma}{Lemma}
\newtheorem{corollary}{Corollary}
\newtheorem{proposition}{Proposition}

\numberwithin{theorem}{section}
\numberwithin{lemma}{section}
\numberwithin{corollary}{theorem}
\numberwithin{proposition}{section}
\numberwithin{remark}{section}

\newcommand{\bs}[1]{\boldsymbol{#1}}

\newcommand{\supp}{{\rm supp}}

\newcommand{\dist}{{\rm dist}}

\title{{\Large {\bf Electric circuit induced by quantum walk }}
\author{
{\small Yusuke Higuchi\footnote{Current address: Department of Mathematics, Faculty of Sciences, Gakushuin University, Mejiro, Toshima-ku, Tokyo, Japan. Email: higuchi@cas.showa-u.ac.jp }}\\
{\scriptsize Mathematics Laboratories, College of Arts and Sciences, Showa University. }\\
{\scriptsize Fujiyoshida, Yamanashi 403-0005, Japan. }\\
{\small Mohamed Sabri\footnote{Email:
sabrimath@dc.tohoku.ac.jp}}\\
{\scriptsize Graduate School of Information Sciences, Tohoku University} \\
{\scriptsize Aoba, Sendai, 980-0845}\\
{\small Etsuo Segawa\footnote{Email: segawa-etsuo-tb@ynu.ac.jp}}\\
{\scriptsize Graduate School of Education Center, Yokohama National University} \\
{\scriptsize Graduate School of Environmental Sciences, Yokohama National University} \\
{\scriptsize Hodogaya, Yokohama 240-8501, Japan.}}
}
\date{\empty }
\begin{document}
\maketitle

\par\noindent
\begin{small}
\par\noindent
{\bf Abstract}. 
We consider the Szegedy walk on graphs adding infinite length tails to a finite internal graph. 
We assume that on these tails, the dynamics is given by the free quantum walk. 
We set the $\ell^\infty$-category initial state so that the internal graph receives time independent input from the tails, say $\bs{\alpha}_{in}$, at every time step. 
We show that the response of the Szegedy walk to the input, 
which is the output, say $\bs{\beta}_{out}$, from the internal graph to the tails in the long time limit, 
is drastically changed depending on the reversibility of the underlying random walk.  
If the underlying random walk is reversible, we have $\bs{\beta}_{out}=\mathrm{Sz}(\bs{m}_{\delta E})\bs{\alpha}_{in}$, 
where the unitary matrix $\mathrm{Sz}(\bs{m}_{\delta E})$ is the reflection matrix to the unit vector $\bs{m}_{\delta E}$ which is determined by the boundary of the internal graph $\delta E$. 
Then the global dynamics so that the internal graph is regarded as one vertex recovers the local dynamics of the Szegedy walk in the long time limit.  
Moreover if the underlying random walk of the Szegedy walk is reversible, 
then we obtain that the stationary state is expressed by a linear combination of the reversible measure 
and the electric current on the electric circuit determined by the internal graph and the random walk's reversible measure. 
On the other hand, if the underlying random walk is not reversible, then the unitary matrix is just a phase flip; 
that is, $\bs{\beta}_{out}=-\bs{\alpha}_{in}$, 
and the stationary state is similar to the current flow but satisfies a different type of the Kirchhoff laws.

\end{small}

\section{Introduction}
An irreducible random walk on a finite graph has a stationary state which is expressed by the eigenvector of the maximal eigenvalue $1$. 
On the other hand, for a quantum walk on a finite graph,  it is not so easy to obtain a stationarity of a quantum walk from a natural initial state. 
The difference between them arises from the following; 
the time evolution operator of a quantum walk is described by some unitary operator on Hilbert space and then its 
spectrum lies on the unit circle in the complex plain. 
However if the graph is infinite, e.g., $d$-dimensional lattice, regular trees, 
then some kinds of convergence theorems of quantum walks have been obtained~\cite{KonnoBook}. 
In \cite{FH1,FH2,HS}, to obtain a stationary state of discrete-time quantum walk, 
quantum walks on the following semi-infinite graph are considered. 
Half lines called ``tails" here, on which the quantum walk is free, are joined to a finite graph named internal graph. 
We set the initial state so that walkers are inputted from these tails to the graph eternally. 
Then the internal graph receives the inputs from the outside at every time step 
while once a quantum walker goes out from the internal graph, she never comes back to the internal graph since the dynamics on tails is free. 
By the balance of these inputs and outputs, we obtain a fixed point of this dynamical system~\cite{HS}. 
The corresponding continuous-time quantum walks have been investigated, in for examples, \cite{FaGoGu,FaGu}. 

In this paper, we further develop \cite{HS} by generalizing Grover walk to the Szegedy walk~\cite{Sze} induced by a random walk on this tailed graph 
to see more fundamental structures of the stationary state.

We find the property of the stationary state is drastically changed depending on the reversibility of the underlying random walk. 
If the underlying random walk is reversible, 
we show that the scattering matrix on the surface of the internal graph is described by the Szegedy matrix determined by only its boundary.
The Szegedy matrix, which is the local quantum coin, determines the local dynamics at each time step and at each vertex (see Section~2 for more detail); it can also describe such global dynamics. 
On the other hand, if the underlying random walk is non-reversible, the scattering matrix on the surface is described by a diagonal matrix. 
Therefore we observe the scattering of a quantum walker towards only the input tails in the long time limit. 
Moreover in the interior of the internal graph 
the stationary state is a convex combination of an electric current and the reversible measure if the underlying random walk is reversible, 
while the stationary state is similar to an electric current but it satisfies a different type of Kirchhoff's law if the underlying random walk is non-reversible. 

Through a simple example, we can observe our model behaves like an $LC$ circuit, in the sense that, as the underlying non-reversible random walk is closer to the reversible one, the total energy of the internal graph in the stationary state increases to infinity~\cite{FW}. See Section~5 for more detail. It is well known that the resonant frequency is normally characterized by the inductance and the capacitance in the circuit; the reversibility in  our setting may be said to be ``resonant frequency", but if the underlying random walk just becomes reversible, the total energy is changed to be finite.  

This paper is organized as follows. 
In Section~2,  the settings of graph and our quantum walk in this paper are prepared. 
In Section~3, we show our main results for the reversible and non-reversible cases, respectively. 
In Section~4, we give the proofs of the theorems. 
Finally Section~5 is devoted to the summary and discussion. 
We show an example to clearly see a difference of stationary states of the quantum walks for the reversible and non-reversible cases,  
which give a motivation for further study on this quantum walk induced by a non-reversible random walk, whose stationary state seems not to be described by a classical dynamics, as a future's problem. 
\section{The setting}
\subsection{Graph setting}
Let $G=(V,A)$ is a symmetric directed graph such that $a\in A$ if and only if  $\bar{a}\in A$, where $\bar{a}$ is the inverse arc of $a$. 
The terminus and the origin vertices of $a\in A$ are denoted by $t(a)$ and $o(a)$, respectively. 
Consider an infinite graph $\tilde{G}=(\tilde{V},\tilde{A})$ constructed of a finite graph $G_0=(V_0,A_0)$ and $r$-tails $\mathbb{P}_1,\dots,\mathbb{P}_r:$ 
	\[ \tilde{G}=G_0 \cup \bigcup_{j=1}^{r}\mathbb{P}_j \mathrm{\;with\;} V(\mathbb{P}_j)\cap V_0 =\{o(\mathbb{P}_j)\}  \]
for $j=1,\dots,r$.
Here each $\mathbb{P}_j$ is the semi-infinite path whose end vertex is $o(\mathbb{P}_j)$. 
Let $\delta A= \{e_1,\dots,e_r\} \subset \cup_{j=1}^r A(\mathbb{P}_j)$ such that $t(e_j)=o(\mathbb{P}_j)$ for $j=1,\dots,r$. 
Assume that the amplitude of the inflow along $e_j$ is $\alpha_j\in \mathbb{C}$. 
\subsection{Underlying random walk and induced Szegedy walk setting}
Throughout this paper, for a discrete set $\Omega$, $\mathbb{C}^\Omega$ is the vector space whose basis are labeled by $\Omega$. 
We set $p: \tilde{A}\to (0,1]$ such that 
	\begin{align*}
        &\sum_{o(a)=u}p(a) = 1,\;\;(u\in \tilde{V}); \\
        &p(a) = 1/2,\;\;(a\in \cup_{j=1}^r A(\mathbb{P}_j)\setminus \{\bar{e}\;;\; e\in\delta A\}). 
        \end{align*}
The probability transition operator $P: \mathbb{C}^{\tilde{V}}\to \mathbb{C}^{\tilde{V}}$ is defined by 
	\[ (Pf)(u)=\sum_{a\in \tilde{A}: t(a)=u} p(a)f(o(a)) \]
for every $u\in \tilde{V}$. 
The cut off of the probability transition operator $P$ with respect to $\delta V=\{o(\mathbb{P}_1),\dots,o(\mathbb{P}_r)\}$; 
$P': \mathbb{C}^{V_0}\to \mathbb{C}^{V_0}$, is denoted by 
	\[ (P'g)(u)=\sum_{a\in A_0: t(a)=u} p(a)f(o(a))  \]
for every $u\in V_0$.         
Let $\tilde{E}$ be the set of unoriented edges of $\tilde{G}$, that is, $|a|=|\bar{a}|\in \tilde{E}$ if and only if $a,\bar{a}\in \tilde{A}$. 
Remark that if the random walk is reversible, that is, there exists $m_V\in \mathbb{R}^{\tilde{V}}\setminus\{\bs{0}\}$ and $m_E\in \mathbb{R}^{\tilde{E}}$, 
such that 
	\begin{equation}\label{eq:DBC} 
	p(a)m_V(o(a))=p(\bar{a})m_V(t(a))=m_{E}(|a|), 
	\end{equation}
then since $p(a)=m_{E}(|a|)/m_V(u)$ for any $a$ with $o(a)=u$ and $m_V(u)=\sum_{o(a)=u}m_E(|a|)$, we have 
	\begin{equation}\label{eq:revRW} 
        [\sqrt{p(a_1)},\dots,\sqrt{p(a_d)}]^\top
        \\=[\sqrt{m_E(|a_1)|/m_V(u)},\dots,\sqrt{m(|a_d|)/m_V(u)}]^\top.  
        \end{equation}
In this paper, we consider the induced quantum walk by this random walk for 
both reversible and non-reversible cases. 

The definition of the induced quantum walk called the Szegedy walk~(cf \cite{Sze}) is as follows. 
The total space of our quantum walk is denoted by $\mathbb{C}^{\tilde{A}}$. 
The time evolution operator of the Szegedy walk on $\mathbb{C}^{\tilde{A}}$ is defined by 
	\[ (U\Psi)(a)=\sum_{b\in\tilde{A}:t(b)=o(a)} (2\sqrt{p(\bar{a})p(b)}-\delta_{\bar{a}}(b))\Psi(b),  \]
for any $\Psi\in \mathbb{C}^{\tilde{A}}$. 
Let $\Psi_n$ be the $n$-th iteration of $U$; that is, $\Psi_{n+1}=U\Psi_n$. 
Put $\{a_1,\dots,a_d\}=\{a\in\tilde{A} \;|\; o(a)=u\}$ for $u\in \tilde{V}$. 
Then we have 
	\begin{equation}\label{eq:localdyamics} 
        \begin{bmatrix} \Psi_{n+1}(a_1) \\ \vdots \\ \Psi_{n+1}(a_d) \end{bmatrix}
        	= \mathrm{Sz}([\sqrt{p(a_1)},\dots,\sqrt{p(a_d)}]^\top) \begin{bmatrix} \Psi_{n}(\bar{a}_1) \\ \vdots \\ \Psi_{n}(\bar{a}_d) \end{bmatrix}. 
        \end{equation}
Here for a unit vector $\bs{u}\in \mathbb{C}^d$, the Szegedy matrix $\mathrm{Sz}(\bs{u})$ is denoted by 
	\[ \mathrm{Sz}(\bs{u})=2 \bs{u}\bs{u}^* -I_{\mathbb{C}^d}, \]
that is, $(\mathrm{Sz}(\bs{u}))_{ij}=2\bs{u}(i)\overline{\bs{u}(j)}-\delta_{ij}$. 
The initial state considered here is give as follows
    \[ \Psi_0(a)=
    \begin{cases}
    \alpha_1 & \text{: $a\in A(\mathbb{P}_1),\;\dist(o(\mathbb{P}_1),t(a))<\dist(o(\mathbb{P}_1),o(a))$,} \\
    \vdots \\
    \alpha_r & \text{: $a\in A(\mathbb{P}_r),\;\dist(o(\mathbb{P}_r),t(a))<\dist(o(\mathbb{P}_r),o(a))$,} \\
    0 & \text{: otherwise.}
    \end{cases} \]
Here $\alpha_1,\dots,\alpha_r\in \mathbb{C}$. 
Note that $\Psi_0$ belongs to no longer the $\ell^2$ category. 
This initial state can be interpreted that we insert the same inflow into the internal graph at the every time step. 
Remark that if a quantum walker goes out to the tails, then it never comes back again to the internal graph due to the free-dynamics on tails; 
such a walker can be regareded as an outflow. 
Then we expect a balance between the inflow and the outflow of the internal graph in the long time limit.
Indeed, we have already obtained the existence of the stationary state of this dynamics in the previous works as follows.  
\begin{theorem}{\rm \cite{HS}}
The stationary state $\Psi_\infty$ uniquely exists, that is, $\exists \lim_{n\to\infty}\Psi_n(a)=\Psi_\infty(a)$. 
The stationary state satisfies 
$(U\Psi_\infty)(a)=\Psi_\infty(a)$. 
\end{theorem}
\subsection{The induced dynamical system}
Let $\chi: \mathbb{C}^{\tilde{A}}\to \mathbb{C}^{A_0}$ be the boundary operator of $A_0$ such that for any $\Psi\in \mathbb{C}^{\tilde{A}}$, 
$(\chi\Psi)(a)=\Psi(a)$ $(a\in A_0)$. 
The adjoint $\chi^*: \mathbb{C}^{A_0}\to \mathbb{C}^{\tilde{A}}$ is described by 
	\[ (\chi^* \psi)(a) = \begin{cases} \psi(a) & \text{: $a\in A_0$,} \\ 0 & \text{: otherwise.} \end{cases} \]
Remark that $\chi\chi^*: \mathbb{C}^{A_0}\to \mathbb{C}^{A_0}$ is the identity operator of $\mathbb{C}^{A_0}$ 
and $\chi^*\chi: \mathbb{C}^{\tilde{A}}\to \mathbb{C}^{\tilde{A}}$ is the projection operator with respect to $A_0$. 
Let $E_{PON}$ be the submatrix of the whole unitary time evolution operator $U$ restricted to $A_0$,  
that is, $E_{PON}=\chi U\chi^*$. 
Putting $\psi_n:=\chi \Psi_n$, we have 
	\begin{align*}
        \psi_n &= \chi \Psi_n=\chi U\Psi_{n-1}\\
               &=\chi U \chi^*\chi \Psi_{n-1}+\chi U(1-\chi^*\chi) \Psi_{n-1}\\
               &= E_{PON}\psi_{n-1}+\chi U\Psi_0
        \end{align*}
for $n\geq 1$. 
Therefore in the internal graph, the dynamics is described by 
	\begin{equation}\label{eq:master_eq} 
        \psi_0=0,\;\; \psi_{n+1}=E_{PON}\psi_n+\rho, 
        \end{equation}
where $\rho$ is the ``input" defined by $\rho:=\chi U\Psi_0$. This is the induced dynamical system. 
\section{Main results}
First we show that 
if the underlying random walk is reversible, 
then the local dynamics of the Szegedy walk denoted by (\ref{eq:localdyamics}) 
is reproduced again in the global dynamics in the long time limit as follows. 
\begin{theorem}\label{thm:scattering}
Let $r$ be the number of tails. 
Assume the underlying random walk is reversible (\ref{eq:DBC}). 
The input and output are denoted by $\bs{\alpha}_{in}=[\Psi_\infty(e_1),\dots,\Psi_\infty(e_r)]^\top\in \mathbb{C}^r$ and 
$\bs{\beta}_{out}=[\Psi_\infty(\bar{e}_1),\dots,\Psi_\infty(\bar{e}_r)]^\top \in \mathbb{C}^r$, respectively. 
Set a unit vector on $\mathbb{C}^r$ by 
	\[\bs{m}_{\delta E}=[ \sqrt{m_E(|e_1|)/m(\delta G_0)}, \dots , \sqrt{m_E(|e_r|)/m(\delta G_0)}]^\top,\] 
where $m(\delta G_0)=\sum_{j=1}^r m_E(|e_j|)$. 
Then we have 
	\[ \bs{\beta}_{out} =  \mathrm{Sz}(\bs{m}_{\delta E}) \; \bs{\alpha}_{in}. \]
\end{theorem}
The scattering for the reversible case is determined by only the {\it surface} of the internal graph. 
In the following theorem, we obtain the information of the stationary state of the {\it interior} of the internal graph as follows. 
\begin{theorem}\label{thm:main}
Let the underlying random walk is reversible (\ref{eq:DBC}). 
Define $\mathrm{j}(\cdot)\in \mathbb{C}^{\tilde{A}}$ by 
	\[ \mathrm{j}(a):=\sqrt{m_E(|a|)}\Psi_\infty(a)-\frac{m_E(|a|)}{\sqrt{m(\delta G_0)}} \langle \bs{m}_{\delta E}, \bs{\alpha}_{in} \rangle. \]
Then $\mathrm{j}(\cdot)$ describes the electric current flow 
of the following electric circuit: 
the conductance is assigned at every edge, and the conductance value at each edge $e$ is given by $m_E(e)$. 
More precisely, $\mathrm{j}(a)$ satisfies Kirchhoff's current and voltage laws: 
\begin{enumerate}
        \item Kirchhoff's current law: 
        \begin{equation}\label{eq:KCL}
        \sum_{t(a)=u}\mathrm{j}(a)  =\sum_{o(a)=u}\mathrm{j}(a)=0,\;\; \mathrm{j}(a)+\mathrm{j}(\bar{a})=0;
        \end{equation}
        \item Kirchhoff's voltage law: 
        \begin{equation}\label{eq:KVL}
            \sum_{k=1}^s \frac{\mathrm{j}(a_k)}{m_E(|a_k|)} = 0\;\; \mathrm{for\;any\;cycle\;}c=(a_1,\dots,a_s).
        \end{equation}
\end{enumerate}

\end{theorem}
To see how quantum walker is distributed on the internal graph in the long time limit, we define 
the relative finding probability at position $u\in V_0$ by 
	\[ \mu_{QW}(u) := \sum_{a\in\tilde{A}: t(a)=u}|\Psi_\infty(a)|^2. \]
Then we obtain the following corollary. 
\begin{corollary}\label{cor:ecrw}
Let the local electric power at vertex $u\in V$ be $w_{EC}(u):=\sum_{t(a)=u}\mathrm{j}(a)v(a)$, where $v(a)$ is the electric potential difference between 
$t(a)$ and $o(a)$, that is, $\mathrm{j}(a)/m_E(|a|)$. Then the relative finding probability of our quantum walk can be expressed by 
sum of the electric power and a reversible measure of the random walk: 
	\[ \mu_{QW}(u)=w_{EC}(u)+m_{RW}(u), \]
where $m_{RW}(u)$ is a reversible measure of the underlying random walk denoted by
	\[ m_{RW}(u) = \frac{|\langle \bs{m}_{\delta E}, \bs{\alpha}_{in} \rangle|^2}{m(\delta G_0)}m_V(u). \]
\end{corollary}
\begin{proof}
It is obtained by a direct computing as follows. 
	\begin{align*}
        \mu_{QW}(u) &= \sum_{t(a)=u}|\Psi_\infty(a)|^2 \\ &=\sum_{t(a)=u}\left|\frac{1}{\sqrt{m_E(|a|)}}\mathrm{j}(a)+\frac{\sqrt{m_E(|a|)}}{\sqrt{m(\delta G_0)}}\langle \bs{m}_{\delta E}, \bs{\alpha}_{in} \rangle\right|^2 \\
               &= \sum_{t(a)=u} \frac{\mathrm{j}^2(a)}{m_E(|a|)}
               	+ \sum_{t(a)=u} \frac{m_E(|a|)}{m(\delta G_0)} |\langle \bs{m}_{\delta E}, \bs{\alpha}_{in} \rangle|^2 \\
               &= w_{EC}(u)+ \frac{|\langle \bs{m}_{\delta E}, \bs{\alpha}_{in} \rangle|^2}{m(\delta G_0)}m_V(u).
        \end{align*}
Here we used Kirchhoff's current law in the second equality and the definition of $w_{EC}$  in the last equality. 
\end{proof}
By the expression of the relative probability, 
we can control our quantum walk ``preference" to electric circuit or random walk 
by adjusting the overlap of input flow $\bs{\alpha}_{in}$ to the conductance of the boundary $\bs{m}_{\delta E}$. 
\begin{corollary}\label{cor:ER}
If $\bs{m}_{\delta E}\bot \bs{\alpha}_{in}$, then $\mu_{QW}(u)=w_{EC}(u)$. 
On the other hand, if $\bs{m}_{\delta E}\parallel \bs{\alpha}_{in}$, then $\mu_{QW}(u)=m_{RW}(u)$. 
\end{corollary}
\begin{proof}
If $\bs{m}_{\delta E}\bot \bs{\alpha}_{in}$, Corollary~\ref{cor:ecrw} implies $\mu_{QW}(u)=w_{EC}(u)$. 
On the other hand, if $\bs{m}_{\delta E}\parallel \bs{\alpha}_{in}$, then by Theorem~\ref{thm:main}, 
the input electric current from $e_j\in \delta A$ can be computed as follows. 
Since $\bs{m}_{\delta E}\parallel \bs{\alpha}_{in}$, the input quantum walk's flow can be described by 
$\alpha_j=c\sqrt{m_E(|e_j|)}$ with some constant $c$. 
Then we have 
	\[ \mathrm{j}(e_j)=\sqrt{m_E(|e_j|)}\alpha_j-\frac{m_E(|e_j|)}{\sqrt{m(\delta G_0)}} \langle \bs{m}_{\delta E}, \bs{\alpha}_{in} \rangle=0,   \]
for any $e_j\in \delta A$. Thus if $\bs{m}_{\delta E}\parallel \bs{\alpha}_{in}$, 
then the electric current flow is not supplied to the internal graph. 
Therefore Kirchhoff's voltage law (\ref{eq:KVL}) implies $\mathrm{j}(a)=0$ for any $a\in \tilde{A}$. 
\end{proof}

On the other hand, if the random walk is non-reversible, 
then we obtain different properties of the stationary state from those of the reversible case; in particular, 
we show that the response is just a phase flip of the input. 
\begin{theorem}\label{thm:non-revcase}
Assume the underlying random walk is non-reversible. 
Let the input and its response be $\bs{\alpha}_{in}$ and $\bs{\beta}'_{put}$, respectively.	
Then we have
	\[ \bs{\beta}_{out}'=-\bs{\alpha}_{in}. \]
Moreover the stationary state $\Psi_\infty$ has the following properties:  
	\begin{align*}
        \sum_{a\in\tilde{A}:t(a)=u} & \sqrt{p(\bar{a})}\Psi_\infty(a) = 0,\;\;(u\in \tilde{V}) \\
        \Psi_\infty(\bar{a}) &= -\Psi_\infty(a)\;\;(a\in\tilde{A}), \\
        \Psi_\infty \in &\ker(1-\chi^* E_{PON} \chi)^\perp.
        \end{align*}
\end{theorem}

In general, we observe the scattering of a quantum walker towards the tails which are not the input tails  in the long time limit for the reversible case. 
On the other hand, for a non-reversible case, we observe the scattering of a quantum walker towards only the input tails.
Then as an inverse problem, observing the scattering way of this quantum walk, we can detect whether the underlying random walk is reversible or not.
\section{Proof of Theorems}
The convergence of $\Psi_n$ is already ensured by \cite{FH1,FH2,HS}, that is, $\Psi_\infty:=\exists \lim_{n\to\infty}\Psi_n$. 
Then $U\Psi_\infty=\Psi_\infty$. 
This is equivalent to 
	\begin{equation}
        \frac{1}{2}\left( \Psi_\infty(a)+\Psi_\infty(\bar{a}) \right) = \sqrt{p(a)}\sum_{t(b)=o(a)} \sqrt{p(\bar{b})} \Psi_\infty(b)
        \end{equation}
for any $a\in \tilde{A}$ by the definition of $U$. 
Putting 
	\[ \rho_E(e) := \frac{1}{2}\sum_{|a|=e} \Psi_\infty(a),\; \rho_V(u) := \sum_{t(b)=u} \sqrt{p(\bar{b})} \Psi_\infty(b)  \] 
for any $e\in \tilde{E}$ and $u\in \tilde{V}$, we have 
	\begin{equation}\label{eqo} \rho_E(|a|)=\sqrt{p(a)}\rho_V(o(a)). \end{equation}
Inserting $\bar{a}$ into the above, we obtain
	\begin{equation}\label{eqt} \rho_E(|a|)=\sqrt{p(\bar{a})}\rho_V(t(a)). \end{equation} 
Comparing (\ref{eqo}) with (\ref{eqt}), we have 
	\begin{equation}\label{eqrev} 
        p(a)|\rho_V(o(a))|^2=p(\bar{a})|\rho_V(t(a))|^2 
        \end{equation}
Then $|\rho_V(\cdot)|^2$ must be the reversible measure (\ref{eq:DBC}) or the null measure. 
\subsection{Proof of Theorem~\ref{thm:scattering}}
In this section, we assume $P$ is reversible, that is, there exist $m_V \in \mathbb{R}^{\tilde{V}}\setminus\{\bs{0}\}$ and $m_E \in \mathbb{R}^{\tilde{E}}$
such that 
	\begin{equation}\label{eq:DBC2} 
	p(a)m_V(o(a)) = p(\bar{a})m_V(t(a)) = m_E(|a|). 
	\end{equation}
Now let us start the proof of Theorem~\ref{thm:scattering}. 
\begin{proof}
From the reversibility (\ref{eq:DBC2}), the probability associated with moving along each arc $a\in A$, $p(a)$, can be expressed by $p(a)=m_E(|a|)/m_V(o(a))$. 
By (\ref{eqrev}), there exists a constant value such that $\rho_V(u)=c\sqrt{m_V(u)}$. 
Remark that if $\rho_V^2$ is the null measure, then $c=0$. 
By (\ref{eqo}), $\rho_E(|a|)=c\sqrt{m_E(|a|)}$ which implies
	\begin{equation} \Psi_{\infty}(\bar{a}) = 2c \sqrt{m_E(|a|)}-\Psi_\infty(a). \end{equation}	
Note that the constant value $c$ depends on the initial state $(\alpha_1,\dots,\alpha_r)$. 

Let us consider the scattering with the initial state inserting the inflow from only a fixed tail $\mathbb{P}_i$, 
that is, $\alpha_j=\delta_{ji}$ for $j=1,\dots,r$. 
Putting the constant value $c$ by $c_i$, we  will determine $c_i$ using the given setting parameters of this model. 
The stationary state with this initial state is denoted by $\Psi^{(i)}\in \mathbb{C}^{\tilde{A}}$ and 
we put $\psi^{(i)}_\infty:=\chi\Psi^{(i)}_\infty \in \mathbb{C}^{A_0}$. 
Let the transmitting value toward the tail $\mathbb{P}_j$ with this initial state 
be rewritten by $\beta_{ji}:=2c_i\sqrt{m_E(|e_j|)}-\delta_{ji}$. 
Then by the stationarity, 
	\[ U(\chi^* \psi^{(i)}_\infty+\delta_{e_j})= \chi^* \psi^{(i)}_\infty + \bs{\beta}^{(i)} \]
holds. Here $\bs{\beta}^{(i)}=\sum_{j=1}^r \beta_{ji}\delta_{e_j}$. 
Putting $\bs{w}^{(i)}:=\chi^* \psi^{(i)}_\infty + \bs{\beta}^{(i)}$, 
we have 
	\begin{align*} 
	\langle \bs{w}^{(i')}, \bs{w}^{(i)} \rangle 
	&= \langle \psi^{(i')}, \psi^{(i)} \rangle + \delta_{ii'} \\
    &= \langle \psi^{(i')}, \psi^{(i)} \rangle + \langle \bs{\beta}^{(i')}, \bs{\beta}^{(i)}\rangle 
    \end{align*}
Thus $\langle \bs{\beta}^{(i')}, \bs{\beta}^{(i)}\rangle=\delta_{i',i}$ which is equivalent to 
	\[  c_i=0\;\mathrm{or}\; c_i=\sqrt{m_E(|e_j|)}/m(\delta G_0).  \]
Here we put $m(\delta G_0):=\sum_{j=1}^r m_E(|e_j|)$. 

Now we will show $c_i\neq 0$ in general. 
Assume $c_i=0$. 
Then $\rho_E(|a|)=\tilde{\rho}_V(u)=0$ holds, which implies $\Psi_\infty^{(i)}(\bar{a})=-\Psi_\infty^{(i)}(a)$ and also
	\[ \sum_{a\in \tilde{A}: t(a)=u}\sqrt{m_E(|a|)}\Psi_\infty^{(i)}(a) =0.  \]
Let us compute $\tau_{G_0}:=\sum_{a\in A_0}\sqrt{m_E(|a|)}\Psi^{(i)}(a)$ from the following two ways:  
	\begin{align*}
        \tau_{G_0}
        	&= \sum_{u\in V_0}\sum_{a\in A_0: t(a)=u}\sqrt{m_E(|a|)}\Psi^{(i)}_\infty(a) \\
                &= -\sum_{j=1}^r \sqrt{m_E(|e_j|)} \beta_{ji}=-1;
        \end{align*}
on the other hand, 
	\begin{align*}
        \tau_{G_0}
        	&= \sum_{u\in V_0}\sum_{a\in A_0: o(a)=u}\sqrt{m_E(|a|)}\Psi^{(i)}_\infty(a) \\
            &= \sum_{u\in V_0}\sum_{a\in A_0: t(a)=u}\sqrt{m_E(|a|)}\Psi^{(i)}_\infty(\bar{a}) \\
                &= -\sum_{u\in V_0}\sum_{a\in A_0: t(a)=u}\sqrt{m_E(|a|)}\Psi^{(i)}_\infty(a)
                = 1.
        \end{align*}
Then the contradiction occurs. 
Therefore we have
\begin{equation}\label{eq:ci}
c_i=\sqrt{m_E(|e_i|)}/m(\delta G_0). 
\end{equation}
%
Then if the inflow is $\bs{\alpha}(j)=\delta_{ij}$ for any $j=1,\dots,r$, we obtain 
	\[ \beta_{ji}=\frac{2}{m(\delta G_0)}\sqrt{m_E(|e_j|)m_E(|e_i|)}-\delta_{ji}. \] 
By the linearity of the time evolution, the stationary state  $\Psi_\infty$ with a general inflow 
is described by a linear combination of $\Psi_\infty^{(i)}$'s. 
Then it holds $\bs{\beta}(j)=\sum_{i=1}^r\beta_{ji}\bs{\alpha}(i)$, which implies the desired conclusion. 
\end{proof}
%
\subsection{Proof of Theorem~\ref{thm:main}}
Theorem~\ref{thm:scattering} shows the information on the stationary state at boundaries of the internal graph. 
Here we will explain a property of the stationary state in the interior of the internal graph. 
\begin{proof}
To show $\mathrm{j}(\cdot)$ is the electric current, it is sufficient to show the following Kirchhoff's laws. 
	\begin{enumerate}
        \item Kirchhoff's current law: 
        \begin{equation*}
        \sum_{t(a)=u}\mathrm{j}(a)  =\sum_{o(a)=u}\mathrm{j}(a)=0,\;\; \mathrm{j}(a)+\mathrm{j}(\bar{a})=0;
        \end{equation*}
        \item Kirchhoff's voltage law: 
        \begin{equation*}
            \sum_{k=1}^s \frac{\mathrm{j}(a_k)}{m_E(|a_k|)} = 0\;\; \mathrm{for\;any\;cycle\;}c=(a_1,\dots,a_s).
        \end{equation*}
        \end{enumerate}
First we show the Kirchhoff's current law. 
Note that $\tilde{\rho}_V(u)$ can be expressed by a linear combination of $c_i$ in (\ref{eq:ci}), that is,  
$\tilde{\rho}_V(u) = \sum_{j=1}^r \alpha_j c_j$. 
Then we have 
	\begin{align*} \tilde{\rho}_V(u)
	&= \frac{1}{m_V(u)} \sum_{t(a)=u} \sqrt{m_E(|a|)}\Psi_\infty(a) \\
    &=  \frac{1}{\sqrt{m(\delta G_0)}} \langle  \bs{m}_{\delta E}, \bs{\alpha}_{in} \rangle, \end{align*}
which implies  
	\begin{equation}\label{eq:veretexavetotal1}
        \sum_{t(a)=u} \sqrt{m_E(|a|)}\Psi_\infty(a) 
        	= \sum_{t(a)=u} \frac{m_E(|a|)}{\sqrt{m(\delta G_0)}} \langle  \bs{m}_{\delta E}, \bs{\alpha}_{in} \rangle. 
	\end{equation}
On the other hand, since $\rho_E(|a|)=\sqrt{m_E(|a|)} \tilde{\rho}_V(o(a))$, 
we have 
	\begin{align*} 
	\rho_E(|a|) 
	 &= \frac{1}{2} (\Psi_\infty(a)+\Psi_\infty(\bar{a})) \\
     &=  \frac{\sqrt{m_E(|a|)}}{\sqrt{m(\delta G_0)}} \langle  \bs{m}_{\delta E}, \bs{\alpha}_{in} \rangle, 
    \end{align*}		
which implies  
	\begin{equation}\label{eq:edgeavetotal}
        \sqrt{m_E(|a|)}(\Psi_\infty(a)+\Psi_\infty(\bar{a}))
        	= \frac{2m_E(|a|)}{\sqrt{m(\delta G_0)}} \langle  \bs{m}_{\delta E}, \bs{\alpha}_{in} \rangle.
        \end{equation}
Therefore putting 
	\[ \mathrm{j}(a):=\sqrt{m_E(|a|)}\Psi_\infty(a)-\frac{m_E(|a|)}{\sqrt{m(\delta G_0)}} \langle \bs{m}_{\delta E}, \bs{\alpha}_{in} \rangle, \]
by (\ref{eq:veretexavetotal1}) and (\ref{eq:edgeavetotal}), 
we obtain Kirchhoff's current law of $\mathrm{j}(a)$:
	\begin{equation}
        \sum_{t(a)=u}\mathrm{j}(a)=\sum_{o(a)=u}\mathrm{j}(a)=0,\;\;\mathrm{j}(a)+\mathrm{j}(\bar{a})=0.
        \end{equation}

Now let us see $\mathrm{j}(\cdot)$ also satisfies the Kirchhoff's voltage law using the following lemma. 
\begin{lemma}\label{lem:cycle}
For any cycle $c=(a_1,a_2,\dots,a_s)$, the induced function in $\mathbb{C}^{\tilde{A}}$ is denoted by
	\[ w_{c}(a)=
        \begin{cases} 
        1/\sqrt{m_E(|a_k|)} & \text{: $a=a_k$ $(k=1,\dots,s)$,}\\
        -1/\sqrt{m_E(|a_k|)} & \text{: $a=\bar{a}_k$ $(k=1,\dots,s)$,}\\
        0 & \text{: otherwise.}
         \end{cases} \]
Then $\langle w_c, \Psi_\infty \rangle=0$ holds for any cycle $c$. 
\end{lemma}
\noindent Proof of Lemma~\ref{lem:cycle}: 
It can be checked that $U w_c=w_c$. Then since the support of $w_c$ is included in $A_0$, $E_{PON}\chi w_c=\chi w_c$ holds. 
By \cite{HS}, the centered eigenspace~\cite{R} of $E_{PON}$, whose absolute value of the eigenvalue is $1$, must be orthogonal to the stationary state. 
Therefore since $\Psi_\infty$ is the fixed point of this dynamical system, 
$\Psi_\infty$ must be orthogonal to these eigenvectors. 
The proof of Lemma~\ref{lem:cycle} is completed. $\square$

By Lemma~\ref{lem:cycle}, it holds that
	\begin{align*}
        \langle w_c, \Psi_\infty \rangle=0 & \Leftrightarrow \sum_{k=1}^s \frac{\Psi_\infty(a_k)-\Psi_\infty(\bar{a}_k)}{\sqrt{m_E(|a_k|)}}=0 \\
        & \Leftrightarrow \sum_{k=1}^s \frac{\mathrm{j}(a_k)}{m_E(|a_k|)}-\frac{\mathrm{j}(\bar{a}_k)}{m_E(|a_k|)}=0 \\
        & \Leftrightarrow \sum_{k=1}^s \frac{\mathrm{j}(a_k)}{m_E(|a_k|)}=0
        \end{align*}
Then we obtain Kirchhoff's voltage law, where the capacitance value assigned at edge $e$ in the electric circuit is $m_E(e)$. 
The proof of Theorem~3.2 is completed.
\end{proof}

\subsection{Proof of Theorem~\ref{thm:non-revcase}}
Here we will give some information on the stationary state  of the quantum walk induced by the non-reversible underlying random walk.  
\begin{proof}
Since $P$ is non-reversible in this section, $\rho_V(u)$ must be $0$ for any $u\in \tilde{V}$ by (\ref{eqrev}). 
Then we have $\rho_E(|a|)=\sqrt{p(a)}\rho_V(o(a))=0$. Thus 
	\begin{equation}\label{eq:nonrev1} 
        \sum_{a\in\tilde{A}:t(a)=u}\sqrt{p(\bar{a})}\Psi_\infty(a)=0 
        \end{equation}
and 
	\begin{align}\label{eq:nonrev2} 
        \Psi_\infty(a)+\Psi_\infty(\bar{a}) &= 0. 
         \end{align} 
Remark that since the underlying random walk is not reversible, then we cannot find further deformation of this properties 
such as connecting electric circuit for the reversible case. 
If we take summation over all the arcs whose origins are $u$ in (\ref{eq:nonrev1}) instead of terminus, then the equality does not always hold, that is, 
	\begin{equation*} 
        \sum_{a\in\tilde{A}:o(a)=u}\sqrt{p(\bar{a})}\Psi_\infty(a) \neq 0. 
        \end{equation*}
because of $p(\bar{a})\neq p(a)$ in general. 
For the reversible case, instead of $\sqrt{p(\bar{a})}$, we could apply the measure on the edge $m_E(|a|)$, 
which is invariant for the inverse arc. This is the critical factor to make a difference between the stationary states of quantum walk
induced by reversible and non-reversible random walks. 
By (\ref{eq:nonrev2}), the relation between the out flow $\beta_j$ and $\alpha_j$ 
can be simply connected by $\beta_j=-\alpha_j$. 
Then the input flow is perfectly reflected with the phase flip. 

The stationary state $\Psi_\infty$ is a solution of the linear equation  $(1-\chi^*E_{PON}\chi)\Psi_{\infty}=\chi^*\rho$.  
By \cite{HS}, the centered eigenspace~\cite{R} of $E_{PON}$, whose absolute value of the eigenvalue is $1$, must be orthogonal to the stationary state. This condition implies that 
$\Psi_\infty \in \ker(1-\chi^* E_{PON}\chi)^{\perp}$, 
Then we have completed the proof. 
\end{proof}

\section{Summary and discussion}
In this paper, we studied the Szegedy walk induced by a random walk on the tailed graph. 
If the underlying random walk is reversible (\ref{eq:DBC}), 
the stationary state is a convex combination of an electric current and the stationary measure 
of the reversible random walk in Theorem~\ref{thm:main}. 
We showed the stationary state depends on the reversibility of the underlying random walk in Theorems~\ref{thm:scattering} and \ref{thm:non-revcase}:
the scattering matrix for the reversible case is described by the Szegedy matrix while
the one for the non-reversible case is just a phase flip. 
Then for the non-reversible case, the non-penetration into the interior of the internal graph may occur. 

Now let us discuss how the quantum walker penetrates into the interior of the internal graph. 
For the reversible case, using the property of the reversible measure, 
we obtain the following Proposition as a consequence of Theorem~\ref{thm:main}. 
This result indicates that a quantum walker either penetrates into the whole arcs of the internal graph or no arcs at all.

\begin{proposition}
Let $\Psi_\infty$ be the stationary state and $\psi_\infty=\chi \Psi_\infty$. 
Then the following two statements for the reversible case with $\langle \bs{m}_{\delta E},\bs{\alpha}_{in} \rangle\neq 0$ are equivalent: 
\begin{enumerate}
\item $\exists a\in A_0$, $\psi_\infty(a)\neq 0$ or $\psi_\infty(\bar{a})\neq 0$;
\item $\forall a\in A_0$, $\psi_\infty(a)\neq 0$ or $\psi_\infty(\bar{a})\neq 0$.
\end{enumerate}
\end{proposition}
An example for $\psi_\infty(a)=0$ for any $a\in A_0$ is as follows. 
Set two tails so that 
a vertex, say $u_*\in V_0$, is connected to the two tails; that is, $t(e_1)=t(e_2)=u_*$. 
Let us consider the initial state $\alpha_1=\sqrt{m_E(|e_2|)}$ and $\alpha_{2}=-\sqrt{m_E(|e_1|)}$, and see $\psi_\infty(a)=0$ for any $a\in A_0$. 
The out source $\rho=\chi U\Psi_0$ to $G_0$ can be computed by 
	\begin{align*} 
	\rho(a) &=\sum_{j=1}^2 2\sqrt{p(a)p(\bar{e}_j)}\\ 
	&= 2 \sqrt{\frac{p(a)}{m_V(o(a))}} \left( m_E(|e_1|)\alpha_1+m_E(|e_2|)\alpha_2 \right)\\
	&=0 
	\end{align*}
for any $a\in A_0$ with $o(a)=u_*$ and then $\rho=0$. 
By (\ref{eq:master_eq}), we have $\psi_n=0$ for any $n\in \mathbb{N}$. 
However eliminating such an exceptional case, 
we see the penetration of quantum walker into {\it all over} the internal graph for the reversible case.

In the next, to see a difference of the penetrations between reversible and non-reversible cases, 
as the internal graph, let us consider the quantum walk 
on a joined graph of the triangle $C_3$ and the finite path $P_k$  of length $k$ ($k\in \mathbb{N}$); $G_0:=C_3*P_k$. 
We add two tails to $G_0$ and  
the initial state is $\alpha_1=1$ and $\alpha_2=0$. 
See Fig.~\ref{fig:one}. 
The moving probability turning clockwise and counterclockwise on $C_3$ be $p$ and $q$ with $1-(p+q)=r>0$; and 
the probability to escape $C_3$ is $r$; the moving probability on every vertex of $P_k$ is $1/2$ except the boundaries. 
See Fig.~\ref{fig:two}. 
We label the arcs of $C_3$ as is depicted in Fig.~\ref{fig:three}. 
\begin{figure}[htbp]
  \begin{center}
   \includegraphics[width=52mm]{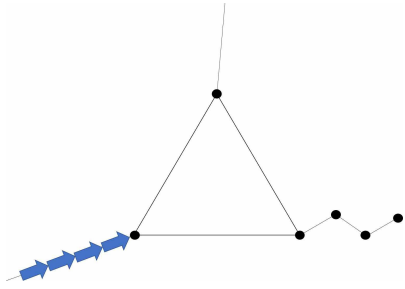}
  \end{center}
  \caption{Initial state}
  \label{fig:one}
\end{figure}
\begin{figure}[htbp]
 \begin{center}
  \includegraphics[width=52mm]{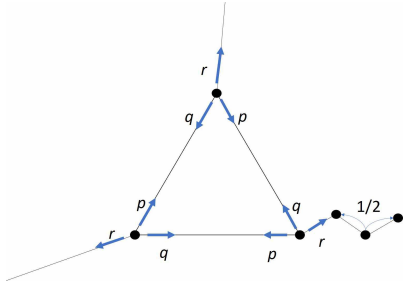}
 \end{center}
  \caption{The underlying random walk}
  \label{fig:two}
\end{figure}
\begin{figure}[htbp]
 \begin{center}
  \includegraphics[width=52mm]{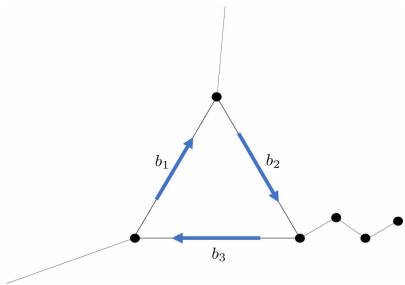}
 \end{center}
  \caption{Labelling of arcs}
  \label{fig:three}
\end{figure}
Remark that $p=q$ if and only if the random walk is reversible (\ref{eq:DBC}), in particular, if $p=q=1/3$, then the induced quantum walk is the Grover walk. 
Now let us consider the cases of $p= q$ and $p\neq q$, respectively. 
\begin{enumerate}
\item $p=q$ case (reversible case) \\
By Theorem~\ref{thm:main}, the stationary state is described by 
	\[ \Psi_\infty(a)=\frac{1}{\sqrt{m_E(|a|)}}\mathrm{j}(a)+\sqrt{\frac{m_E(|a|)}{m(\delta G_0)}}\langle \bs{m}_{\delta E},\bs{\alpha}_{in} \rangle. \]
Because of $p=q=(1-r)/2$, then the conductances are proportional to  
	\[ m_E(e) = \begin{cases} 1-r & \text{: $e\in E(C_3)$; }\\  2r  & \text{: otherwise}. \end{cases} \] 
Then $m(\delta G_0)=2r +2r =4r$, and $\bs{m}_{\delta E}=[\sqrt{1/2},\sqrt{1/2}]^\top$, which implies $\langle \bs{m}_{\delta E},\bs{\alpha}_{in} \rangle=1/\sqrt{2}$. 
Therefore the stationary state in this case is reduced to
	\begin{equation}\label{eq:a} 
        \Psi_\infty(a)=\frac{1}{\sqrt{m_E(|a|)}}\mathrm{j}(a)+\sqrt{m_E(|a|)}\frac{1}{\sqrt{8r}}. 
        \end{equation}
Then we need to compute the following current flow on the graph. 

By (\ref{eq:a}), the value of the current flow from the outside into the internal graph; $I$, is 
	\[ I=\sqrt{2r}\cdot 1 -\frac{2r}{\sqrt{4r}}\sqrt{\frac{2r}{4r}}=\sqrt{\frac{r}{2}}. \]
Note that by the Kirchhoff's current law, the current does not exist on $P_3$. 
Then putting $\mathrm{j}(b_1)=I_1$, we have $\mathrm{j}(\bar{b_2})=\mathrm{j}(\bar{b_3})=I-I_1$. 
Moreover putting $R=(1-r)^{-1}$, by Kirchhoff's voltage law, we have $I_1R=(I-I_1)R+(I-I_1)R$ which implies $I_1=(2/3)I$.
Thus the stationary state restricted to the internal graph for the reversible case is described by 
	\begin{align*}
        \psi_\infty(b_1) &= \frac{3+r}{6 \sqrt{2} \sqrt{r(1-r)}}; \;\; \psi_\infty(\bar{b}_1) = \frac{3-7 r}{6 \sqrt{2} \sqrt{r(1-r)}}; \\
        \psi_\infty(b_2) &= \frac{3-5 r}{6 \sqrt{2} \sqrt{r(1-r)}};\;\;  \psi_\infty(\bar{b}_2) = \frac{3-r}{6 \sqrt{2} \sqrt{r(1-r)}}; \\
        \psi_\infty(b_3) &= \frac{3-5 r}{6 \sqrt{2} \sqrt{r(1-r)}};\;\;  \psi_\infty(\bar{b}_3) = \frac{3-r}{6 \sqrt{2} \sqrt{r(1-r)}}; \\
        \psi_\infty(b) &= 1/2 \;\;\mathrm{for\;any}\;b\in A(P_k).
        \end{align*}
Here we remark that the value of $\psi_\infty$ for each arc does not depend on the constant multiple of the revesible measure. 
In particular, if $r=1/3$ which induces the Grover walk case, then 
	\begin{align*} 
        \psi_\infty(b_1) &= 5/6; \;\;  \psi_\infty(\bar{b}_1) = 1/6; \\
        \psi_\infty(b_2) &= 1/3; \;\; \psi_\infty(\bar{b}_2) = 2/3; \\
        \psi_\infty(b_3) &= 1/3; \;\; \psi_\infty(\bar{b}_3) = 2/3; \\
        \psi_\infty(b) &= 1/2 \;\;\mathrm{for\;any}\;b\in A(P_k). 
        \end{align*} 
Moreover by Theorem~\ref{thm:scattering}, since the scattering matrix in this case is reduced to 
	\[ \mathrm{Sz}(\bs{m}_{\delta E})=\begin{bmatrix} 0 & 1 \\ 1 & 0 \end{bmatrix}, \]
then the perfectly transmitting occurs. 
\item $p \neq q$ case (non-reversible) \\
Focusing on $t(b_1)$, by Theorem~\ref{thm:non-revcase}, we have 
	\[ \sqrt{q}\Psi_\infty(b_1)+\sqrt{p}\Psi_\infty(\bar{b}_2)=0. \]
Putting $\Psi_\infty(b_1)=:\beta$, we have $\Psi_\infty(\bar{b}_2)=-(\sqrt{q/p})\beta$. 
By Theorem~\ref{thm:non-revcase}, we have $\Psi_\infty(b_2)=(\sqrt{q/p})\beta$. 
Let us see $\Psi_\infty(a)=0$ for all $a\in A(P_k)=\{e_1,\dots,e_k,\bar{e}_1,\dots,\bar{e}_k\}$ in the following, where
$t(b_2)=o(e_1),t(e_1)=o(e_2),\dots, t(e_{k-1})=o(e_k)$. 
Focusing on the leaf of the path; that is, $t(e_k)$, by Theorem~\ref{thm:non-revcase}, we have $\Psi_\infty(e_k) \times 1=0$.
Thus $\Psi_\infty(\bar{e}_k)=-\Psi_\infty(e_k)=0$.  
Using this argument recursively, we obtain $\Psi_\infty(e_{k-1})=\Psi_\infty(\bar{e}_{k-1})=0,\dots, \Psi_\infty(e_{k-1})=\Psi_\infty(\bar{e}_{k-1})=0$. 
Next, focusing on the right neighbor of $t(b_1)$; $t(b_2)$, by Theorem~\ref{thm:non-revcase} again, we have 
	\[ \sqrt{q}\Psi_\infty(b_2)+\sqrt{p}\Psi_\infty(\bar{b}_3)+\sqrt{r}\Psi_\infty(\bar{e}_1)=0. \]
Since $\Psi_\infty(\bar{e}_1)=0$, we have $\Psi_\infty(b_3)=(\sqrt{q^2/p^2})\beta$. 
Finally, focusing on the right neighbor of $t(b_2)$; $t(b_3)$, by Theorem~\ref{thm:non-revcase}, we have 
	\[ \sqrt{q}\Psi_\infty(b_3)-\sqrt{p}\beta+1\cdot \sqrt{r}=0. \]
Then we obtain 
	\[ \beta=\frac{\sqrt{p^2r}}{p^{3/2}-q^{3/2}}. \]
Thus the stationary state restricted to the internal graph for the non-reversible case is described by 
	\begin{align*}
        \psi_\infty(b_1) &= -\psi_\infty(\bar{b}_1)=\frac{\sqrt{rp^2}}{p^{3/2}-q^{3/2}}; \\
        \psi_\infty(b_2) &= -\psi_\infty(\bar{b}_2)=\frac{\sqrt{rpq}}{p^{3/2}-q^{3/2}}; \\
        \psi_\infty(b_3) &= -\psi_\infty(\bar{b}_3)=\frac{\sqrt{rq^2}}{p^{3/2}-q^{3/2}}; \\
        \psi_\infty(b) &= 0 \;\;\mathrm{for\;any}\;b\in A(P_k).
        \end{align*}
Therefore the quantum walker {\it partially} penetrates into the internal graph; this is due to the existence of a leaf. 
\end{enumerate}
The perfect reflection occurs for the non-reversible case by Theorem~\ref{thm:non-revcase}, while the perfectly transmitting occurs for the reversible case. 

Let $M(G_0):=\sum_{t(a)\in V_0}|\Psi_\infty(a)|^2$ be the mass of the internal graph from the view point of our quantum walk. 
We summarize the stationary state on $C_3*P_k$ for each case in the following table. 
\begin{center}
\begin{tabular}{r|c|c|l}
 & $\supp(\psi_\infty)$ & $M(G_0)$ & \qquad\quad Scattering \\ \hline
Rev. $(p=q)$ & $A_0$ & $M_{rev}$ & perfect transmission \\ 
N-Rev. $(p\neq q)$ & $A(C_3)\subset A_0$ & $M_{nonrev}$ & perfect reflection  
\end{tabular}
\end{center}
\noindent \\
Here 
	\begin{equation}\label{eq:revcase} 
        M_{rev}=\left(-\frac{17}{12}+\frac{2}{3(1-r)}+\frac{3}{4r}+\frac{k}{2}\right)+1
        \end{equation}
and 
	\begin{equation}\label{eq:nonrevcase} 
        M_{nonrev}=\frac{2r(p^2+pq+q^2)}{(p^{3/2}-q^{3/2})^2}+1, 
        \end{equation}
in particular, if $p=q=r=1/3$ corresponding to the Grover walk case,  
	\[ M(G_0)=M_{rev}=11/6 + k/2+1.\]
For the function $M_{rev}(r)$ on $0\leq r\leq 1$, 
we can observe that $M_{rev}(0)=M_{rev}(1)=0$ by the definition of Szegedy walk, while 
$\lim_{r\downarrow 0} M_{rev}(r)=\lim_{r\uparrow 1} M_{rev}(r)=\infty$. 
On the other hand, let  us regard $M_{nonrev}$ as the function $M_{nonrev}(\epsilon;r)$ of $\epsilon=|p-q|$ for a fixed $0<r<1$. 
Then, if $\epsilon=0$,  $M_{nonrev}$ corresponds to the reversible case. However, interestingly, 
by (\ref{eq:revcase}) and (\ref{eq:nonrevcase}), 
$M_{nonrev}$ is not continuously accumulated to $M_{rev}(r)<\infty$ as $\epsilon\downarrow 0$; in fact, $\lim_{\epsilon \downarrow 0}M_{nonrev}(\epsilon;r)=\infty$. 
To study a quantum walk induced by a non-reversible random walk is one of the interesting future's problem since 
this study investigates such a ``phase transition" of the quantum walk in more detail. 



\noindent \\
\noindent \\
\noindent {\bf Acknowledgments}
The authors would like to express their sincere gratitude to Professor Hiroshi Ogura for his valuable suggestions.
YuH's work was supported in part by Japan Society for the Promotion of Science Grant-in-Aid for Scientific Research 
(C)~25400208, (C)~18K03401 and (A)~15H02055.
E.S. acknowledges financial supports from Japan Society for the Promotion of Science Grant-in-Aid for Scientific Research (C) 19K03616, and Research Origin for Dressed Photon.



\begin{small}
\bibliographystyle{jplain}

\begin{thebibliography}{99}

\bibitem{FaGoGu}
Farhi, J.,  Goldstone, S., Gutmann, S.: 
A  Quantum  Algorithm  for  the  Hamiltonian  NAND Tree,
Theory of Computing {\bf 4} (2008) pp.169-190.

\bibitem{FaGu}
Farhi, E., Gutmann, S.:
Quantum computation and decision trees, 
Phys. Rev. A {\bf 58} (1998) pp.915--928. 

\bibitem{FH1}
Feldman, E., Hillery, M.:
Quantum walks on graphs and quantum scattering theory, 
Coding Theory and Quantum Computing, edited by D. Evans, J. Holt, C. Jones, K. Klintworth, B. Parshall, O. Pfister, and H. Ward, 
Contemporary Mathematics  {\bf 381} (2005) pp.71--96. 

\bibitem{FH2}
Feldman, E., Hillery, M.:
Modifying quantum walks: A scattering theory approach
Journal of Physics A: Mathematical and Theoretical {\bf 40} (2007) 11319.  

\bibitem{FW}
Fukushima, Y., Wada, T.: 
A discrete transmission line model for discrete-time quantum walks, 
Interdisciplinary information sciences {\bf 23} (2017) 87-93. 

\bibitem{HS}
Higuchi, Yu., Segawa, E.: 
Dynamical system induced by quantum walk, Journal of Physics A: Mathematical and Theoretical {\bf 52} (2019) 395202.

\bibitem{KonnoBook} 
Konno, N.: 
Quantum Walks. 
In: Lecture Notes in Mathematics: {\bf 1954} (2008) pp.309--452, Springer-Verlag, Heidelberg.

\bibitem{R}
Robinson, M.: 
Dynamical Systems: 
Stability, Symbolic dynamics, and Chaos, 
CRC Press (1995). 

\bibitem{Sze}
M. Szegedy, Quantum speed-up of Markov chain based algorithms, 
Proc. 45th IEEE Symposium on Foundations of Computer Science (2004), 32--41.

\end{thebibliography}

\end{small}

\end{document}